\documentclass[12pt,reqno]{amsart}

\usepackage[latin1]{inputenc}
\usepackage{amsmath}
\usepackage{amsfonts}
\usepackage{amssymb}
\usepackage{graphics}
\usepackage{enumerate}
\usepackage{amssymb,amsmath,amsthm,amscd,latexsym,verbatim,graphicx,amsfonts}

\usepackage{amscd}
\usepackage{amsmath}
\usepackage{amssymb}
\usepackage{amsthm}
\usepackage{latexsym}
\usepackage{verbatim}

\theoremstyle{plain}
\newtheorem{theorem}{Theorem}

\theoremstyle{definition}

\theoremstyle{remark}

\newcommand{\nri}{n\rightarrow\infty}

\newcommand{\bbR}{\mathbb{R}}
\newcommand{\bbC}{\mathbb{C}}

\newcommand{\bbD}{\mathbb{D}}

\newcommand{\bbN}{\mathbb{N}}

\newcommand{\mcc}{\mathcal{C}}

\newcommand{\mcl}{\mathcal{L}}
\newcommand{\mcm}{\mathcal{M}}

\DeclareMathOperator*{\Real}{Re}
\DeclareMathOperator*{\Imag}{Im}

\title[]{Orthogonal Polynomials on the Unit Circle, Mutually Unbiased Bases, and Balanced States}
\author[]{Graeme Reinhart and Brian Simanek}
\date{}

\begin{document}
\maketitle

\begin{abstract}
Two interesting phenomena for the construction of quantum states are that of mutually unbiased bases and that of balanced states.  We explore a constructive approach to each phenomenon that involves orthogonal polynomials on the unit circle.  In the case of mutually unbiased bases, we show that this approach does not produce such bases.  In the case of balanced states, we provide examples of pairs of orthonormal bases and states that are balanced with respect to them.  We also consider extensions of these ideas to the infinite dimensional setting.
\end{abstract}

\vspace{4mm}


\vspace{2mm}


\vspace{2mm}

\normalsize

\section{Introduction}\label{intro}

In $\bbC^n$, two orthonormal bases $\{\phi_j\}_{j=1}^n$ and $\{\psi_j\}_{j=1}^n$ are said to be \textit{mutually unbiased} if
\begin{equation}\label{dotc}
|\langle\phi_j,\psi_k\rangle|=1/\sqrt{n}
\end{equation}
for every $j,k=1,\ldots,n$.  For example, each pair of bases from among the set
\[
\left\{\begin{pmatrix}1\\0\end{pmatrix},\begin{pmatrix}0\\1\end{pmatrix}\right\},\qquad\qquad\left\{\begin{pmatrix}1/\sqrt{2}\\1/\sqrt{2}\end{pmatrix},\begin{pmatrix}1/\sqrt{2}\\-1/\sqrt{2}\end{pmatrix}\right\},\qquad\qquad\left\{\begin{pmatrix}1/\sqrt{2}\\i/\sqrt{2}\end{pmatrix},\begin{pmatrix}1/\sqrt{2}\\-i/\sqrt{2}\end{pmatrix}\right\}
\]
forms a pair of mutually unbiased bases for $\bbC^2$.  In this case we say that these three bases are mutually unbiased.  It is known that in $\bbC^n$, one can find no more than $n+1$ mutually unbiased bases \cite[Theorem 3.5]{BBRV14} and it is known that this bound is sharp if $n$ is a power of a prime number \cite{ASSW14,BBRV14,WF89}.  If $n$ is not a power of a prime number, then it is unknown if this bound is sharp.

The proof of the sharpness of the bound $n+1$ in prime powered dimensions is constructive (see \cite{WF89}, and then later \cite{BBRV14}) and depends on the existence of a finite field with size $n$.  Ongoing work to resolve the question of sharpness of this bound in other dimensions has motivated searches for new methods of constructing mutually unbiased bases.  This work is motivated by precisely that problem.  We will consider families of orthonormal bases derived from families of orthogonal polynomials on the unit circle and for which we can evaluate the dot products \eqref{dotc} in closed form.  The calculation will show that such bases are never mutually unbiased.

Related to the notion of mutually unbiased bases is a state that is \textit{balanced} with respect to two orthonormal bases (see \cite{ASSW14}).  If $\{\phi_j\}_{j=1}^n$ and $\{\psi_j\}_{j=1}^n$ are two orthonormal bases for $\bbC^n$, then a vector $\chi$ is balanced with respect to these bases if the collections of numbers
\[
\left\{|\langle\chi,\phi_j\rangle|\right\}_{j=1}^n,\qquad\qquad\mbox{and}\qquad\qquad\left\{|\langle\chi,\psi_j\rangle|\right\}_{j=1}^n
\]
are the same.  Balanced states were initially considered in \cite{ASSW14} with respect to mutually unbiased bases, but the notion of a balanced state can be applied to any pair of orthonormal bases.  Given two orthonormal bases, they can be made to form the columns of unitary matrices $U$ and $V$.  A state $\vec{w}$ is balanced with respect to these bases if
\[
U^*\vec{w}=PV^*\vec{w},
\]
where $P$ is a matrix of all zeros except each row and column has exactly one non-zero entry, which is on the unit circle.  We see that $\vec{w}$ must be an eigenvector for $VP^*U^*$ with eigenvalue $1$, and the matrix $P$ can always be adjusted so that $1$ is an eigenvalue of this matrix.  We will consider a different approach to finding balanced states that allows us to extend this idea to infinite dimensional spaces.  Indeed, if the operators in question act on infinite dimensional spaces, then $VP^*U^*$ need not have any eigenvalues.

The idea behind the terminology for mutually unbiased bases and balanced states comes from quantum mechanics.  Recall that an orthonormal basis $\{\phi_j\}_{j=1}^n$ for $\bbC^n$ determines a measurement that can be applied to any state vector (that is, any normalized vector) $\chi$ in $\bbC^n$.  The possible outcomes of the measurement are the basis vectors and the probability of outcome $\phi_j$ is $|\langle\chi,\phi_j\rangle|^2$.  Furthermore, if the outcome $\phi_j$ is observed, then the state vector emerges from the measurement in state $\phi_j$.  Thus, if $\{\phi_j\}_{j=1}^n$ and $\{\psi_j\}_{j=1}^n$ are two orthonormal bases for $\bbC^n$, then to say they are mutually unbiased means that if the measurements are applied in succession, the probability distribution on the outcomes of the second measurement is uniform.  The terminology for a balanced state can be understood similarly.  If $\{\phi_j\}_{j=1}^n$ and $\{\psi_j\}_{j=1}^n$ are two orthonormal bases for $\bbC^n$, which determine two measurements $M_{\phi}$ and $M_{\psi}$ on a state vector in $\bbC^n$, then to say that $\chi$ is balanced with respect to these bases means that if one prepares many systems in state $\chi$ and performs the same measurement ($M_{\phi}$ or $M_{\psi}$) on each such state, then one cannot get any information about which measurement was performed just by looking at the distribution of the measurement outcomes.

We will explore mutually unbiased bases and balanced states using orthogonal polynomials on the unit circle (OPUC).  Given a Borel probability measure $\mu$ whose support is an infinite subset of the unit circle $\partial\bbD$ in the complex plane, one can consider the corresponding sequence of orthonormal polynomials $\{\varphi_m\}_{m=0}^{\infty}$.  By evaluating these polynomials at certain points of $\partial\bbD$, one can obtain an orthonormal basis for $\bbC^n$.  If $\lambda\in\partial\bbD$, then one can form the Alexandrov Family of probability measures associated to $\mu$, which we denote by $\{\mu^{(\lambda)}\}_{|\lambda|=1}$.  The relationship between $\mu$ and each $\mu^{(\lambda)}$ is complicated (see Section \ref{alex} below) and is most easily understood through the orthonormal polynomials.  Each $\mu^{(\lambda)}$ comes with its own sequence of orthonormal polynomials $\{\varphi_m^{(\lambda)}\}_{m=0}^{\infty}$ and again one can obtain an orthonormal basis for $\bbC^n$.  We can evaluate the inner product between any two vectors in these orthonormal basis using a new mixed Christoffel-Darboux formula (see Theorem \ref{mixalex} bellow) and we will see that these basis are never mutually unbiased.

There is also interest in considering mutually unbiased bases in infinite dimensional spaces (see \cite{WW08}).  While it is easy to see that the definition of mutually unbiased basis cannot be directly adapted to the infinite dimensional setting, one can replace vectors by generalized vectors to obtain a meaningful analog.    One can similarly discuss balanced states in the mutually unbiased setting.  If we again permit the use of generalized vectors, we can find two uncountable collections of vectors $\{\phi_w\}_{w\in I}$ and $\{\psi_w\}_{w\in I}$ and a state $\chi$ so that the sets
\[
\left\{|\langle\chi,\phi_w\rangle|\right\}_{w\in I},\qquad\qquad\mbox{and}\qquad\qquad\left\{|\langle\chi,\psi_w\rangle|\right\}_{w\in I}
\]
are the same.  The construction will make heavy use of a specific example of OPUC where precise formulas are available for the polynomials and other related quantities.

The next section will present the necessary background about OPUC.  Section \ref{newmix} will include our first new result, which is the general mixed Christoffel-Darboux formula mentioned earlier.  Section \ref{app} will apply this formula to show that the bases in question are not mutually unbiased, to construct a balanced state with respect to two orthonormal bases, and to present analogous results in the infinite dimensional case.

\section{OPUC Background}\label{mixed}

In this section we will present some background information about orthogonal polynomials on the unit circle.  We will confine our attention to the specific formulas and theorems that we will use elsewhere in this paper.  The interested reader can consult the references \cite{OPUC1,OPUC2} for additional information.  For the rest of this paper, we will use $\mu$ to denote a Borel probability measure whose support is an infinite subset of $\partial\bbD$.

\subsection{The Szeg\H{o} Recursion}\label{recur}

One can perform Gram-Schmidt on the sequence $\{1,z,z^2,\ldots\}$ in the space $L^2(\mu)$ to obtain the sequence of orthonormal polynomials $\{\varphi_n\}_{n=0}^{\infty}$.  If we divide each $\varphi_n$ by its (positive) leading coefficient, then we obtain a new collection of monic orthogonal polynomials $\{\Phi_n\}_{n=0}^{\infty}$.  The Szeg\H{o} Recursion states that there is a sequence of numbers $\{\alpha_n\}_{n=0}^{\infty}$ in the unit disk $\bbD$ such that
\begin{align*}
\Phi_{n+1}(z)&=z\Phi_n(z)-\bar{\alpha}_n\phi_n^*(z)\\
\Phi_{n+1}^*(z)&=\Phi_n^*(z)-\alpha_nz\Phi_n(z),
\end{align*}
where $\Phi_n^*(z)=z^n\overline{\Phi_n(1/\bar{z})}$ are the reversed polynomials (see \cite[Theorem 1.5.2]{OPUC1}).  It is easy to see from these formulas and the fact that $\Phi_n^*(0)=1$ that $-\bar{\alpha}_n=\Phi_{n+1}(0)$.  The sequence $\{\alpha_n\}_{n=0}^{\infty}$ is often called the sequence of Verblunsky coefficients associated to the measure $\mu$, but one can find other names in the literature as well.  The orthonormal polynomials also satisfy a form of the Szeg\H{o} Recursion, as can be seen by the formula
\[
\varphi_n(z)=\frac{\Phi_n(z)}{\prod_{j=0}^{n-1}\sqrt{1-|\alpha_j|^2}}.
\]
We will occasionally write $\varphi_n(z;\mu)$ if we need to emphasize the dependence on the measure $\mu$.

\subsection{Verblunsky's Theorem}\label{verb}

We have just seen that to each measure $\mu$ as in Section \ref{recur}, one can associate a sequence $\{\alpha_n\}_{n=0}^{\infty}$ in the unit disk.  Verblunsky's Theorem states that the converse is also true.  Specifically, it states that to any sequence of complex numbers $\{\alpha_n\}_{n=0}^{\infty}$ in the unit disk, there is a measure on the unit circle whose corresponding monic orthogonal polynomials satisfy $-\bar{\alpha}_n=\Phi_{n+1}(0)$ (see \cite[Theorem 1.7.11]{OPUC1}).  Due to this theorem, it is common to refer to a measure on the unit circle by referring to its corresponding sequence of Verblunsky coefficients.

\subsection{Alexandrov Measures}\label{alex}

Given a measure $\mu$ with corresponding Verblunsky coefficients $\{\alpha_n\}_{n=0}^{\infty}$ and a choice of $\lambda\in\partial\bbD$, one can define a new measure $\mu^{(\lambda)}$ by defining its Verblunsky coefficients via the formula
\[
\alpha_n(\mu^{(\lambda)})=\lambda\alpha_n(\mu).
\]
Thus, the Verblunsky coefficients for $\mu^{(\lambda)}$ are just rotated versions of the Verblunsky coefficients for $\mu$ and so they are still in $\bbD$.  We caution the reader that the measure $\mu^{(\lambda)}$ is not just a rotated version of the measure $\mu$ (see Section \ref{rotate}).  The relationship is more subtle and is easiest to understand by the relationship between the Verblunsky coefficients.

\subsection{The CD Kernel and Christoffel Functions}\label{CD}

In the space $L^2(\mu)$, the orthogonal projection onto $\mbox{span}\{1,z,\ldots,z^n\}$ is given by integration against the kernel
\[
K_n(z,w)=\sum_{j=0}^n\varphi_j(z)\overline{\varphi_j(w)},
\]
which is sometimes called the Christoffel-Darboux kernel.  The Christoffel-Darboux Formula is an expression for this same kernel that does not involve a sum (see \cite[Theorem 2.2.7]{OPUC1}).  Theorem \ref{mixalex} below will present a generalized version of this formula.  We will occasionally write $K_n(z,w;\mu)$ if we need to emphasize the dependence on the measure $\mu$.

Another important use of the Christoffel-Darboux kernel is its relevance to Christoffel functions.  For any $z\in\bbC$ and $n\in\bbN$, define
\[
\lambda_n(\zeta)=\inf\left\{\int|P(z)|^2d\mu(z):P(z)=1,\deg(P)\leq n\right\}.
\]
The infimum defining $\lambda_n$ is in fact a minimum and the unique extremal polynomial is
\[
\frac{K_n(z,\zeta)}{K_n(\zeta,\zeta)}.
\]
From this, it immediately follows that $\lambda_n(\zeta)=K_n(\zeta,\zeta)^{-1}$ (see \cite[Section 2.2]{OPUC1}). As a corollary, we deduce that $\mu(\{\zeta\})>0$ if and only if $\{\varphi_j(\zeta)\}_{j=0}^{\infty}\in\ell^2$ and in fact
\begin{equation}\label{massform}
\mu(\{\zeta\})=\left(\sum_{j=0}^{\infty}|\varphi(\zeta)|^2\right)^{-1}
\end{equation}
(see \cite[Section 2.2]{OPUC1}).

\subsection{Rotated Measures}\label{rotate}

One can easily calculate the orthogonal polynomials and Verblunsky coefficients of a rotated version of the measure $\mu$.  It is easy to see that if $\sigma\in\partial\bbD$, then $\{\Phi_n(\sigma z)\}_{n=0}^{\infty}$ are orthogonal with respect to the measure $\mu$ that has been rotated by $\sigma$ on the unit circle.  However, $\Phi_n(\sigma z)$ is not monic, so to make it monic, we multiply by $\bar{\sigma}^n$.  Then, using the fact that $\alpha_n=-\overline{\Phi_{n+1}(0)}$, we find that the Verblunsky coefficients of the rotated measure are given by
\[
\alpha_n(\sigma\mu)=\bar{\sigma}^{n+1}\alpha_n(\mu).
\]
From this we also calculate
\[
\varphi_n(z;\sigma\mu)=\bar{\sigma}^n\varphi_n(\sigma z;\mu)
\]
from which it follows that
\[
K_n(z,w;\sigma\mu)=K_n(\sigma z,\sigma w;\mu).
\]

\subsection{Paraorthogonal Polynomials}\label{POPUC}

Related to the orthogonal polynomials on the unit circle are the paraorthogonal polynomials on the unit circle (POPUC).  For any $\beta\in\partial\bbD$, we can define the POPUC $\Phi_n(z;\beta)$ by
\[
\Phi_n(z;\beta)=z\Phi_{n-1}(z)-\bar{\beta}\Phi_{n-1}^*(z).
\]
Thus, the POPUC are defined in terms of the OPUC.  By changing $\alpha_{n-1}\in\bbD$ to a parameter $\beta\in\partial\bbD$, the polynomial changes in a way that has a predictable effect on the zeros.  The following facts are true of the zeros of POPUC:
\begin{itemize}
\item The zeros of $\Phi_n(z;\beta)$ are all on the unit circle $\partial\bbD$.
\item The zeros of $\Phi_n(z;\beta)$ are all simple.
\item If $\beta,\tau\in\partial\bbD$ and $\beta\neq\tau$, then the zeros of $\Phi_n(z;\beta)$ interlace on the unit circle.
\end{itemize}
All of these facts can be proven using the observation that $z\Phi_{n-1}(z)/\Phi_{n-1}^*(z)$ is a Blaschke product.

\subsection{Quadrature Measures}\label{quad}

Given a measure $\mu$ on $\partial\bbD$ as above and a natural number $n\in\bbN$, one can find a measure $\nu_n$ supported on $n$ points in $\partial\bbD$ and for which
\[
\int z^jd\mu=\int z^jd\nu_n,\qquad\qquad j=-n,-n+1,\ldots,n-1,n.
\]
The measure $\nu_n$ is called a quadrature measure for $\mu$.  The measure $\nu_n$ is not unique and there is a one-parameter family of such measures that can be constructed explicitly.  Indeed, for any $\beta\in\partial\bbD$, one can consider the zeros of the polynomial $\Phi_{n}(z;\beta)$, which we denote by $\{z_j\}_{j=1}^n$.  We have already seen that $|z_j|=1$ for all $j$.  Then, one defines
\[
\nu_n=\sum_{j=1}^nw_j\delta_{z_j},
\]
where
\begin{equation}\label{weights}
w_j=K_{n-1}(z_j,z_j)^{-1}
\end{equation}
(see \cite[Theorem 2.2.12]{OPUC1}).  From this formula, we deduce the non-trivial corollary that
\[
\sum_{j=1}^nK_{n-1}(z_j,z_j)^{-1}=1.
\]

\subsection{CMV Matrices}\label{cmv}

Given a sequence of Verblunsky coefficients $\{\alpha_n\}_{n=0}^{\infty}$, one can form a sequence of $2\times2$ unitary matrices $\{\Theta_n\}_{n=0}^{\infty}$ as follows:
\[
\Theta_n=\begin{pmatrix}
\bar{\alpha}_n & \sqrt{1-|\alpha_n|^2}\\
\sqrt{1-|\alpha_n|^2} & -\alpha_n
\end{pmatrix}.
\]
One can then form two infinite unitary matrices $\mcl$ and $\mcm$ by the formulas
\[
\mcl=\Theta_0\oplus\Theta_2\oplus\Theta_4\oplus\cdots,\qquad\qquad\mcm=1\oplus\Theta_1\oplus\Theta_3\oplus\cdots,
\]
where the $1$ that starts the formula for $\mcm$ is a $1\times1$ identity matrix.  The CMV matrix corresponding to the sequence $\{\alpha_n\}_{n=0}^{\infty}$ is the unitary matrix $\mcc=\mcl\mcm$.

The matrix $\mcc$ is cyclic with respect to the vector $\vec{e}_1$ and the spectral measure of $\mcc$ with respect to this vector is the measure $\mu$ whose Verblunsky coefficients are $\{\alpha_n\}_{n=0}^{\infty}$ (see \cite[Section 1.7]{OPUC1} and \cite[Section 4.2]{OPUC1}).  Furthermore, the characteristic polynomial of the principle $n\times n$ submatrix of $\mcc$ is $\Phi_n(z;\mu)$.  It is easy to see that if one replaces $\alpha_{n-1}$ by some $\beta\in\partial\bbD$, then the principle $n\times n$ submatrix of $\mcc$ becomes unitary, which is why all the zeros of paraorthogonal polynomials are on the unit circle (they are eigenvalues of a unitary matrix).

\subsection{Sieved Measures}\label{sieve}

If $\mu$ is a measure on $\partial\bbD$, we can think of it as a measure on $[0,2\pi)$ by means of the mapping $e^{i\theta}\rightarrow\theta$.  If $p\in\bbN$, we can $p$-sieve such a measure by placing $p$ copies of the measure on the intervals $[2j\pi/p,2(j+1)\pi/p)$ for $j=0,1,\ldots,p-1$ and each having total measure $1/p$.  If the Verblunsky coefficients for $\mu$ are $\{\alpha_n\}_{n=0}^{\infty}$, then the Verblunsky coefficients $\{\tilde{\alpha}_n\}_{n=0}^{\infty}$ for the sieved measure are
\[
\tilde{\alpha}_{kp+m}=\begin{cases}
\alpha_k \qquad \mbox{if}\qquad &m=p-1\\
0 & \mbox{otherwise}
\end{cases}
\]
(see \cite[Section 1.6]{OPUC1}).  From this, it follows that if $\{\Phi_n\}_{n=0}^{\infty}$ are the monic orthogonal polynomials for $\mu$ and $\{\tilde{\Phi}_n\}_{n=0}^{\infty}$ are the monic orthogonal polynomials for the sieved measure, then
\begin{equation}\label{sivpoly}
\tilde{\Phi}_{kp+m}(z)=z^m\Phi_k(z^p)
\end{equation}
with a similar relation for the orthonormal polynomials.

\section{A Mixed CD Formula}\label{newmix}

Let us retain the notation for the orthonormal polynomials $\{\varphi_n\}_{n=0}^{\infty}$ and all the orthonormal polynomials in the Alexandrov family $\{\varphi_n^{(\lambda)}\}_{n=0}^{\infty}$, where $|\lambda|=1$.  Or first result is a generalized version of the Mixed CD formula.

\begin{theorem}\label{mixalex}
If $z\bar{w}\neq1$, then
\begin{align*}
\sum_{j=0}^n\overline{\varphi_j(w)}\varphi_j^{(\lambda)}(z)&=\frac{1-\bar{\lambda}+\bar{\lambda}\overline{\varphi_{n+1}^*(w)}\varphi_{n+1}^{(\lambda)*}(z)-\overline{\varphi_{n+1}(w)}\varphi_{n+1}^{(\lambda)}(z)}{1-z\bar{w}}\\
&=\frac{1-\bar{\lambda}-z\bar{w}\overline{\varphi_n(w)}\varphi_n^{(\lambda)}(z)+\bar{\lambda}\overline{\varphi_n^*(w)}\varphi_n^{(\lambda)*}(z)}{1-z\bar{w}}
\end{align*}
\end{theorem}

Theorem \ref{mixalex} in the case $\lambda=1$ is the usual CD formula \cite[Theorem 2.2.7]{OPUC1} and Theorem \ref{mixalex} when $\lambda=-1$ is what is usually called the Mixed CD formula \cite[Theorem 3.2.3]{OPUC1}.  We will omit the details of the proof because the proof proceeds exactly as the proof of \cite[Theorem 3.2.3]{OPUC1} once we notice that
\[
\label{matrec}
\begin{pmatrix}
\bar{\lambda}\varphi_{n+1}^{(\lambda)*}(z)\\
\varphi_{n+1}^{(\lambda)}(z)
\end{pmatrix}=
M_n(z)
\begin{pmatrix}
\bar{\lambda}\varphi_{n}^{(\lambda)*}(z)\\
\varphi_{n}^{(\lambda)}(z)
\end{pmatrix}
\]
where
\[
M_n(z)=\frac{1}{\rho_n}\begin{pmatrix}
1 & -\alpha_nz\\
-\bar{\alpha}_n & z
\end{pmatrix}
.
\]

\section{Applications}\label{app}

Now we are ready to discuss some applications of Theorem \ref{mixalex}.

\subsection{Mutually Unbiased Bases}\label{mubapp}

 Suppose $|\beta|=1$ and $\Phi_{n+1}(z_k;\beta)=0$ for $k=1,\ldots,n+1$.  Then $z_k\varphi_n(z_k)=\bar{\beta}\varphi_n^*(z_k)$.  Similarly, if $|\lambda|=1$ and $\Phi_{n+1}^{(\lambda)}(w_k;\lambda\beta)=0$ for $k=1,\ldots,n+1$, then $w_k\varphi_n^{(\lambda)}(w_k)=\bar{\lambda}\bar{\beta}\varphi_n^{(\lambda)*}(w_k)$.  If we make these substitutions in Theorem \ref{mixalex}, then we find
\begin{equation}\label{dotp}
\sum_{j=0}^n\overline{\varphi_j(z_k)}\varphi_j^{(\lambda)}(w_m)=\frac{1-\bar{\lambda}}{1-\bar{z}_kw_m}.
\end{equation}
Notice that this implies $z_k\neq w_m$ (when $\lambda\neq1$), for otherwise the right-hand side of \eqref{dotp} would be infinitte while the left-hand side would not.  For ease of notation, let us define the vector
\[
\vec{\varphi}(z)=(\varphi_0(z),\varphi_1(z),\ldots,\varphi_n(z))^T,
\]
and similarly define $\vec{\varphi}^{(\lambda)}(z)$.  Equation \eqref{dotp} with $\lambda=1$ implies
\[
\left\{\vec{\varphi}(z_k)\right\}_{k=1}^{n+1}
\]
is an orthogonal basis for $\bbC^{n+1}$.  Hence, the same is true for
\[
\left\{\vec{\varphi}^{(\lambda)}(w_k)\right\}_{k=1}^{n+1}.
\]
Notice that these are not orthonormal bases and the required normalization constants are precisely the weights on the associated quadrature measures from Section \ref{quad}.  Define
\[
\vec{\eta}(z)=\frac{\vec{\varphi}(z)}{\sqrt{K_n(z,z)}},\qquad\qquad \vec{\eta}^{(\lambda)}(z)=\frac{\vec{\varphi}^{(\lambda)}(z)}{\sqrt{K_n^{(\lambda)}(z,z)}}.
\]

\begin{theorem}\label{nomub}
The orthonormal bases $\{\vec{\eta}(z_k)\}$ and $\{\vec{\eta}^{(\lambda)}(w_k)\}$ are not mutually unbiased.
\end{theorem}

\begin{proof}
By Theorem \ref{mixalex}, if the given bases were mutually unbiased, then we would have
\[
\frac{1}{K_n(z_k,z_k)K_n^{(\lambda)}(w_m,w_m)}\cdot\frac{|1-\bar{\lambda}|^2}{|z_k-w_m|^2}=\frac{1}{n+1}
\]
for all $k,m=1,2,\ldots,n+1$ (we used the fact that $|z_k|=|w_m|=1$).  This is a collection of $(n+1)^2$ equations.  Let us consider all those equations in which $k=1$ and use the fact that
\[
\sum_{k=1}^{n+1}\frac{1}{K_n(z_k,z_k)}=\sum_{m=1}^{n+1}\frac{1}{K_n^{(\lambda)}(w_m,w_m)}=1
\]
(see Section \ref{quad}). This subset of equations implies
\[
\frac{1}{K_n(z_1,z_1)}=\frac{\sum_{j=1}^{n+1}|z_1-w_j|^2}{(n+1)|1-\bar{\lambda}|^2},\qquad\qquad\frac{1}{K_n^{(\lambda)}(w_m,w_m)}=\frac{|z_1-w_m|^2}{\sum_{j=1}^{n+1}|z_1-w_j|^2}.
\]
This means
\begin{equation}\label{K1}
\frac{1}{K_n^{(\lambda)}(w_m,w_m)}=\frac{|z_1-w_m|^2K_n(z_1,z_1)}{(n+1)|1-\bar{\lambda}|^2}.
\end{equation}
We can repeat this same analysis, but with $k=p$ for any other value $p\in\{2,\ldots,n+1\}$.  This gives us
\begin{equation}\label{K2}
\frac{1}{K_n^{(\lambda)}(w_m,w_m)}=\frac{|z_p-w_m|^2K_n(z_p,z_p)}{(n+1)|1-\bar{\lambda}|^2}
\end{equation}
for all $m=1,\ldots,n+1$.  Comparing \eqref{K1} nad \eqref{K2} shows
\begin{equation}\label{K3}
\frac{|z_1-w_m|^2}{|z_p-w_m|^2}=\frac{K_n(z_1,z_1)}{K_n(z_p,z_p)}
\end{equation}
and this relation is true for all $m=1,2,\ldots,n+1$.  We conclude that
\begin{equation}\label{K4}
\frac{|z_q-w_m|^2}{|z_p-w_m|^2}=\frac{|z_q-w_k|^2}{|z_p-w_k|^2}
\end{equation}
for all $k,m,p,q\in\{1,\ldots,n+1\}$.

We can rewrite this as
\[
|z_q-w_m|^2|z_p-w_k|^2=|z_q-w_k|^2|z_p-w_m|^2.
\]
This expression simplifies to
\begin{equation}\label{reim}
\Real[(z_q-z_p)(\bar{w}_k-\bar{w}_m)]=\Imag[z_p\bar{z}_q]\Imag[\bar{w}_kw_m]
\end{equation}
By rotating the measure if necessary (see Section \ref{rotate}), we may assume without loss of generality that $z_q=\bar{z}_p$.  With this simplification, we may rewrite \eqref{reim} as
\[
\frac{\Imag[w_k-w_m]}{\Imag[w_k\bar{w}_m]}=\frac{\Imag[z_p^2]}{2\Imag[z_p]}
\]
and this must hold for all $k,m,p\in\{1,\ldots,n+1\}$.  Thus
\[
\frac{\Imag[w_k-w_m]}{\Imag[w_k\bar{w}_m]}=\frac{\Imag[w_j-w_m]}{\Imag[w_j\bar{w}_m]}
\]
for all $k,m,j\in\{1,\ldots,n+1\}$.  If we let $w_t=e^{i\phi_t}$, then this reduces to
\[
\frac{\sin(\phi_k)-\sin(\phi_m)}{\sin(\phi_k-\phi_m)}=\frac{\sin(\phi_j)-\sin(\phi_m)}{\sin(\phi_j-\phi_m)}.
\]

Notice that for each $y\in[0,2\pi]$, 
\[
\frac{\sin(x)-\sin(y)}{\sin(x-y)}
\]
is an injective function of $x$.  Therefore, $\phi_j=\phi_k$ for all $j,k\in\{1,\ldots,n+1\}$ and hence $w_j=w_k$.  This gives us the desired contradiction to the assumption that the bases are mutually unbiased.
\end{proof}

We can appeal to additional results in the theory of OPUC to see why we might expect Theorem \ref{nomub} to be true in many cases.  Suppose we write
\[
d\mu(t)=\nu(t)\frac{dt}{2\pi}+d\mu_s(t),
\]
where we have interpreted $\mu$ as a measure on $[0,2\pi]$ and $\mu_s$ is singular with respect to Lebesgue measure (and suppose we have a similar decomposition for each $\mu^{(\lambda)}$).  Then the main result of \cite{MNT91} shows that under some relatively weak assumptions on the measure $\mu$, it holds that
\[
\lim_{\nri}\frac{n}{K_n(e^{it},e^{it})}=\nu(t)
\]
for Lebesgue almost every $t\in[0,2\pi]$ (we used Section \ref{CD} here).  Thus, for large $n$, we would expect
\[
\vec{\eta}(z_k)\cdot\vec{\eta}^{(\lambda)}(w_m)\approx\frac{(1-\bar{\lambda})\sqrt{\nu(z_k)\nu^{(\lambda)}(w_m)}}{n(1-\bar{z}_kw_m)}.
\]
If $\mu$ is such that $\nu$ and $\nu^{(\lambda)}$ are both bounded above, then we see that $|\vec{\eta}(z_k)\cdot\vec{\eta}^{(\lambda)}(w_m)|$ will be biased in favor of those pairs $(z_k,w_m)$ that are close together.

The proof in \cite{BBRV14} of the existence of mutually unbiased bases in prime power dimension relies on finding large collections of mutually commuting unitary matrices.  In light of Theorem \ref{nomub}, the following result should not be surprising.

\begin{theorem}\label{cmvcomm}
Suppose $\mcc_1$ and $\mcc_2$ are two CMV matrices.  If $\mcc_1\mcc_2=\mcc_2\mcc_1$, then $\mcc_1=\mcc_2$.
\end{theorem}

\begin{proof}
Let $\{\alpha_n\}_{n=0}^{\infty}$ be the Verblunsky coefficients that determine the entries of $\mcc_1$ and let $\{a_n\}_{n=0}^{\infty}$ be the Verblunsky coefficients that determine the entries of $\mcc_2$.  Let us start by making the assumption that all $\{\alpha_n\}_{n=0}^{\infty}$ are non-zero.  By equating the diagonal entries of $\mcc_1\mcc_2$ and $\mcc_2\mcc_1$, one quickly sees that $\alpha_1=a_1$ and $\bar{\alpha}_{j+1}a_{j-1}=\bar{a}_{j+1}\alpha_{j-1}$ for all $j\in\bbN$.  This immediately implies $\alpha_k=a_k$ for all odd $k$, and we will complete the proof by showing that $\alpha_0=a_0$.

To see this, start by equating the $(3,1)$-entries of $\mcc_1\mcc_2$ and $\mcc_2\mcc_1$ to see that $\alpha_2r_0=a_2\rho_0$, where $\rho_j=\sqrt{1-|\alpha_j|^2}$ and $r_j=\sqrt{1-|a_j|^2}$.  Now if we equate the $(2,1)$-entries of these matrices and use the fact that $\alpha_1=a_1$, then we see
\[
\bar{\alpha}_1(\alpha_0r_0-a_0\rho_0)=\rho_0\bar{a}_0-r_0\bar{\alpha}_0.
\]
If we can divide both sides by $\alpha_0r_0-a_0\rho_0$, then we get $|\alpha_1|=1$, which is a contradiction.  It follows that $\alpha_0r_0=a_0\rho_0$.  Now we can combine the equalities $\alpha_0r_0=a_0\rho_0$ and $a_2\rho_0=\alpha_2r_0$ to obtain $a_0\alpha_2=a_2\alpha_0$ and multiply this by the equality $\bar{a}_0\alpha_2=\bar{\alpha}_0a_2$ to obtain $\alpha_2^2|a_0|^2=a_2^2|\alpha_0|^2$.  This leads to
\[
\alpha_2^2r_0^2-a_2^2=a_2^2\rho_0^2-a_2^2.
\]
Now we use the fact that $\alpha_2r_0=a_2\rho_0$ to obtain $\alpha_2^2=a_2^2$.  This implies $\alpha_2=a_2$ since $\alpha_2r_0=a_2\rho_0$.  It follows that $\alpha_0=a_0$ as desired.

Now we need to consider the possibility that some of the Verblunsky coefficients are $0$.  In this setting, we can still apply all of the above reasoning to conclude that $\alpha_1=a_1$, $\alpha_2=a_2$, and $\bar{\alpha}_{j+1}a_{j-1}=\bar{a}_{j+1}\alpha_{j-1}$ for all $j\in\bbN$.  If $a_0\neq\alpha_0$, then it must be that $\alpha_2=a_2=0$.  By comparing the $(1,3)$-entries of the matrices $\mcc_1\mcc_2$ and $\mcc_2\mcc_1$, using the fact that $a_1=\alpha_1$ and the fact that $\alpha_0r_0=a_0\rho_0$ (which is still true), then we conclude that $\alpha_1a_0=\alpha_1\alpha_0$.  Thus, either $\alpha_0=a_0$ or $\alpha_1=0$.  In the latter case, by comparing the $(4,3)$-entries of $\mcc_1\mcc_2$ and $\mcc_2\mcc_1$ we also conclude $a_0=\alpha_0$.

We finish the proof by induction.  Suppose $\alpha_j=a_j$ for $j=0,\ldots,n$ and for contradiction that $\alpha_{n+1}\neq a_{n+1}$.  Then since $\bar{\alpha}_{n+1}a_{n-1}=\bar{a}_{n+1}\alpha_{n-1}$, we conclude that $a_{n-1}=\alpha_{n-1}=0$.  By comparing the $(n+2,n)$ or $(n,n+2)$-entries of $\mcc_1\mcc_2$ and $\mcc_2\mcc_1$, we find that
\[
\bar{\alpha}_{n+1}a_{n-2}=\bar{a}_{n+1}\alpha_{n-2}
\]
and hence $a_{n-2}=\alpha_{n-2}=0$.  Then by comparing the $(n+1,n)$ or $(n,n+1)$-entries of $\mcc_1\mcc_2$ and $\mcc_2\mcc_1$, we find that $\alpha_{n+1}=a_{n+1}$, which gives the desired contradiction.
\end{proof}

In the case of finite matrices, when two unitary matrices commute, they can be simultaneously diagonalized.  An analogous result for infinite dimensional unitary operators could potentially yield a shorter proof of Theorem \ref{cmvcomm}.  However, it is not clear that such an analog exists, though one might try to reason as in \cite[Exercise VII.4]{RS1} to find such a result.

\subsection{Balanced States}\label{balapp}

Suppose that $\mu$ is such that all the Verblunsky coefficients are real numbers.  Then $\varphi_j(1)\in\bbR$ for all $j\in\bbN$.  If $|\lambda|=1$ and $\lambda\not\in\bbR$, then all the Verblunsky coefficients of $\mu^{(\lambda)}$ are the complex conjugates of the Verblunsky coefficients of $\mu^{(\bar{\lambda})}$.  Therefore,
\[
\overline{\varphi_j^{(\lambda)}(z)}=\varphi_j^{(\bar{\lambda})}(\bar{z}).
\]
In particular, the roots of the POPUC $\Phi_{n+1}^{(\lambda)}(z;\lambda)$ are the complex conjugates of the roots of $\Phi_{n+1}^{(\bar{\lambda})}(z;\bar{\lambda})$.

Let us denote the roots of $\Phi_{n+1}^{(\lambda)}(z;\lambda)$ by $\{w_j\}_{j=1}^{n+1}$.  Then by Theorem \ref{mixalex} we calculate
\begin{align*}
\vec{\eta}(1)\cdot\vec{\eta}^{(\lambda)}(w_m)&=\frac{1}{\sqrt{K_n(1,1)K_n^{(\lambda)}(w_m,w_m)}}\sum_{j=0}^n\varphi_j(1)\varphi_j^{(\lambda)}(w_m)\\
&=\frac{(1-\bar{\lambda})}{(1-w_m)\sqrt{K_n(1,1)K_n^{(\lambda)}(w_m,w_m)}}
\end{align*}
We can apply a similar calculation to find a formula for $\vec{\eta}(1)\cdot\vec{\eta}^{(\bar{\lambda})}(\bar{w}_m)$ and deduce the following result.

\begin{theorem}\label{balance}
Suppose $\mu$ is such that all its Verblunsky coefficients are real and $\lambda\in\partial\bbD\setminus\{-1,1\}$.  Let $\{w_m\}_{m=1}^{n+1}$ be the zeros of $\Phi_{n+1}^{(\lambda)}(z;\lambda)$.  Then the state $\vec{\eta}(1)$ is balanced with respect to the orthonormal bases
\[
\left\{\vec{\eta}^{(\lambda)}(w_m)\right\}_{m=1}^{n+1},\qquad\mbox{and}\qquad \left\{\vec{\eta}^{(\bar{\lambda})}(\bar{w}_m)\right\}_{m=1}^{n+1}.
\]
In other words, the sets
\[
\left\{|\vec{\eta}(1)\cdot\vec{\eta}^{(\lambda)}(w_m)|\right\}_{m=1}^{n+1},\qquad\mbox{and}\qquad \left\{|\vec{\eta}(1)\cdot\vec{\eta}^{(\bar{\lambda})}(\bar{w}_m)|\right\}_{m=1}^{n+1}
\]
are the same.
\end{theorem}

The content of Theorem \ref{balance} is not especially deep in that it relies only on the fact that taking a complex conjugate does not change the absolute value.  However, it is useful in that it leads us to an interesting example in an infinite dimensional space, which we explore in the next section.

\subsection{Infinite Dimensions}\label{infinite}

Now we will consider the analog of balanced states in infinite dimensional vector spaces.  MUBs have been previously considered in infinite dimensional spaces \cite{WW08}, and it was noted there that a perfect analog is not possible.  One must relax the conditions somewhat to find an appropriate analog and we will do the same to define balanced states.  For our calculations, we will consider vectors in $\ell^{\infty}(\bbN)$, but still use the inner product in $\ell^2(\bbN)$.  Notice that if $\vec{v}_1\in\ell^{\infty}(\bbN)$, then in order to ensure that the inner product $\vec{v}_1\cdot\vec{v}_2$ converges, one would ordinarily need $\vec{v}_2\in\ell^1(\bbN)$.  We will construct a balanced state for which the necessary inner products converge for an uncountable collection of vectors that are analogous to the basis vectors of Section \ref{balapp}, but the state vector is not in $\ell^1(\bbN)$.

In analogy with our earlier notation, let us define
\[
\vec{\varphi}(z)=(\varphi_0(z),\varphi_1(z),\varphi_2(z),\ldots)^T.
\]
By Theorem \ref{mixalex}, we can still define the inner product $\vec{\varphi}(z)\cdot\vec{\varphi}^{(\lambda)}(w)$ if
\begin{equation}\label{zerofin}
\lim_{\nri}\varphi_n(z)=0,\qquad\qquad\limsup_{\nri}|\varphi_n^{(\lambda)}(w)|<\infty.
\end{equation}
In this case, we have (if $|z|=|w|=1$)
\begin{align*}
\vec{\varphi}(z)\cdot\vec{\varphi}^{(\lambda)}(w)&=\lim_{\nri}\sum_{j=0}^n\overline{\varphi_j(z)}\varphi_j^{(\lambda)}(w)\\
&=\lim_{\nri}\frac{1-\bar{\lambda}-w\bar{z}\overline{\varphi_n(z)}\varphi_n^{(\lambda)}(w)+\bar{\lambda}\overline{\varphi_n^*(z)}\varphi_n^{(\lambda)*}(w)}{1-w\bar{z}}\\
&=\frac{1-\bar{\lambda}}{1-w\bar{z}}.
\end{align*}
We used the fact that $|\varphi_n(z)|=|\varphi_n^*(z)|$ when $|z|=1$.

To apply this idea, we need to find a measure $\mu$ for which \eqref{zerofin} holds for some $z,w,\lambda\in\partial\bbD$.  We can find an example of such a measure in \cite{Ranga10}.  The measure is of the form $s(\theta)d\theta$, where
\begin{equation}\label{example1}
s(\theta)=\tau[\sin(\theta/2)]^{2b},\qquad\qquad0\leq\theta\leq2\pi,
\end{equation}
and $-1/2<b<0$ (see \cite[Theorem 4.1]{Ranga10}).  The number $\tau$ is a normalization constant that makes this a probability measure (and depends on $b$).  The Verblunsky coefficients for this measure are given by
\[
\alpha_n=-\frac{b}{b+n+1}
\]
(see \cite[Theorem 3.1]{Ranga10}).

The first thing to observe about this sequence is that $\sum|\alpha_n|^2<\infty$.  Therefore, in order to verify the conditions \eqref{zerofin}, it suffices to consider the analogous conditions for the monic orthogonal polynomials.  Secondly, we notice that this is a sequence of bounded variation, meaning
\[
\sum_{n=0}^{\infty}|\alpha_{n+1}-\alpha_n|<\infty.
\]
Since the same condition holds if we replace $\alpha_n$ by $\lambda\alpha_n$, we conclude that the Verblunsky coefficients for $\mu^{(\lambda)}$ are also of bounded variation.  Therefore, we may apply \cite[Theorem 10.12.5]{OPUC2} to conclude that if $w\in\partial\bbD\setminus\{1\}$, then
\[
\limsup_{\nri}|\varphi_n^{(\lambda)}(w)|<\infty.
\]
It also follows from \cite[Equation 3.5]{Ranga10} that
\begin{equation}\label{notl1}
\Phi_n(1)=\frac{(2b+1)_n}{(b+1)_n}=\prod_{j=1}^n\left(1+\frac{b}{b+j}\right),
\end{equation}
which tends to zero as $\nri$ (see \cite[page 155]{SS03}).  Here we used the Pochhammer symbol
\[
(c)_n=c(c+1)\cdots(c+n-1),\qquad\qquad n\in\bbN
\]
and $(c)_0=1$.  We have therefore verified the conditions \eqref{zerofin} for this example.  We summarize our findings in the following theorem.

\begin{theorem}\label{balance2}
Suppose $\mu$ is given by $s(\theta)d\theta$ with $s$ as in \eqref{example1} and $\lambda\in\partial\bbD\setminus\{-1,1\}$.  Then the state $\vec{\varphi}(1)$ is balanced with respect to the sets
\[
\left\{\vec{\varphi}^{(\lambda)}(w)\right\}_{w\in\partial\bbD\setminus\{1\}},\qquad\mbox{and}\qquad \left\{\vec{\varphi}^{(\bar{\lambda})}(w)\right\}_{w\in\partial\bbD\setminus\{1\}}.
\]
In other words, the sets
\[
\left\{|\vec{\varphi}(1)\cdot\vec{\varphi}^{(\lambda)}(w)|\right\}_{w\in\partial\bbD\setminus\{1\}},\qquad\mbox{and}\qquad \left\{|\vec{\varphi}(1)\cdot\vec{\varphi}^{(\bar{\lambda})}(w)|\right\}_{w\in\partial\bbD\setminus\{1\}}
\]
are the same.
\end{theorem}

Notice that part of the content of Theorem \ref{balance2} again only depends on the fact that complex conjugation does not change the absolute value.  However, it is significant because equation \eqref{notl1} shows that $\vec{\varphi}(1)\not\in\ell^1(\bbN)$.  That means the balanced state we have constructed has a well-defined inner product with an uncountable collection of $\ell^{\infty}(\bbN)$ vectors, but the vector itself is not in $\ell^1(\bbN)$.

Theorem \ref{balance2} finds a single state that is balanced with respect to two uncountable collections.  For any $p\in\bbN$, one can find examples with $p$ balanced states by appealing to sieved measures (see Section \ref{sieve}).  Indeed, if we take the example from Theorem \ref{balance2} and $p$-sieve it, then the formulas from Section \ref{sieve} imply that conditions \eqref{zerofin} are satisfied for any $z$ for which $z^p=1$ and $w\in\partial\bbD$ for which $w^p\neq1$.  This leads us to the following result.

\begin{theorem}\label{balance3}
Suppose $\nu$ is given by $s(\theta)d\theta$ with $s$ as in \eqref{example1}.  Suppose $p\in\bbN$ and $\mu$ is the $p$-sieved version of the measure $\nu$. Suppose $\lambda\in\partial\bbD\setminus\{-1,1\}$.  If $u^p=1$, then the state $\vec{\varphi}(u)$ is balanced with respect to the sets
\[
\left\{\vec{\varphi}^{(\lambda)}(w)\right\}_{w\in\partial\bbD\setminus\{z:z^p=1\}},\qquad\mbox{and}\qquad \left\{\vec{\varphi}^{(\bar{\lambda})}(w)\right\}_{w\in\partial\bbD\setminus\{z:z^p=1\}}.
\]
In other words, the sets
\[
\left\{|\vec{\varphi}(u)\cdot\vec{\varphi}^{(\lambda)}(w)|\right\}_{w\in\partial\bbD\setminus\{z:z^p=1\}},\qquad\mbox{and}\qquad \left\{|\vec{\varphi}(u)\cdot\vec{\varphi}^{(\bar{\lambda})}(w)|\right\}_{w\in\partial\bbD\setminus\{z:z^p=1\}}
\]
are the same.
\end{theorem}

\end{document}